\documentclass[reprint, 10pt, floatfix,aps,prd,twocolumn,a4paper,superscriptaddress,nofootinbib,notitlepage]{revtex4-1}
\makeatletter
\renewcommand\thefigure{\@arabic\c@figure}
\renewcommand\fnum@figure{\figurename~\thefigure}
\newcommand\addfigref[1]{\renewcommand\fnum@figure{\figurename~\thefigure~#1}}
\newcommand\nofigref{\renewcommand\fnum@figure{\figurename~\thefigure}}
\makeatother

\usepackage[utf8]{inputenc}
\usepackage[T1]{fontenc} 
\usepackage{paralist}
\usepackage[breaklinks]{hyperref}
\usepackage[english]{babel}
\usepackage{amsopn,amsthm,dsfont,amssymb}
\usepackage{mathrsfs}
\usepackage{cancel} 
\usepackage{soul} 
\usepackage[normalem]{ulem}

\usepackage{mathtools}
\mathtoolsset{showonlyrefs=true}
\usepackage{tikz}
\usetikzlibrary{positioning,calc,arrows,decorations.pathmorphing,shapes}
\usepackage{adjustbox}
\usepackage[caption=false]{subfig}

\newcommand{\ie}{{\it i.e.}}

\newtheorem{theorem}{Theorem}
\newtheorem{lemma}{Lemma}
\newtheorem{corollary}{Corollary}
\newtheorem{definition}{Definition}

\newtheorem*{blank*}{}
\newcommand{\source}{\circledR}
\newcommand{\sink}{\circledS}
\newcommand{\mouth}[2]{
	\foreach\i in {0,0.2,...,1}{
		\draw[fill=black,opacity=\i*0.4] #1 ellipse ({(1-\i)*0.4cm} and {(1-\i)*0.2cm}) node (#2) {};
	}
}

\begin{document}
\title{Reversible time travel with freedom of choice}
\author{\"Amin Baumeler}
\affiliation{Institute for Quantum Optics and Quantum Information (IQOQI), Austrian Academy of Sciences, Boltzmanngasse 3, Vienna A-1090, Austria}
\affiliation{Facolt\`{a} indipendente di Gandria, Lunga scala, 6978 Gandria, Switzerland}
\author{Fabio Costa}
\affiliation{Centre for Engineered Quantum Systems, School of Mathematics and Physics,\\The University of Queensland, St.\ Lucia, QLD 4072, Australia}
\author{Timothy C.\ Ralph}
\affiliation{Centre for Quantum Computation and Communication Technology, School of Mathematics and Physics,\\The University of Queensland, St.\ Lucia, QLD 4072, Australia}
\author{Stefan Wolf}
\affiliation{Faculty of Informatics, Universit\`{a} della Svizzera italiana, Via G.\ Buffi 13, 6900 Lugano, Switzerland}
\affiliation{Facolt\`{a} indipendente di Gandria, Lunga scala, 6978 Gandria, Switzerland}
\author{Magdalena Zych}
\affiliation{Centre for Engineered Quantum Systems, School of Mathematics and Physics,\\The University of Queensland, St.\ Lucia, QLD 4072, Australia}
\begin{abstract}
	General relativity allows for the existence of closed time-like curves, along which a material object could travel back in time and interact with its past self. This possibility raises the question whether certain initial conditions, or more generally local operations, lead to inconsistencies and should thus be forbidden.
	Here we consider the most general deterministic dynamics connecting classical degrees of freedom defined on a set of bounded space-time regions, requiring that it is compatible with {\em arbitrary\/} operations performed in the local regions. We find that any such dynamics can be realised through reversible interactions. We further find that consistency with local operations is compatible with non-trivial time travel: Three parties can interact in such a way to be all both {\em in the future and in the past\/} of each other, while being {\em free\/} to perform arbitrary local operations.
\end{abstract}

\maketitle

\section{Introduction}

One of the most baffling aspects of general relativity is that certain solutions to the Einstein equations contain closed time-like curves (CTCs)~\cite{Lanczos:1924kn, Godel:1949eb, Taub:1951, NUT:1963, Kerr:1963, Tipler:1974iw, Griffiths:2009exactBook}, where an event can be both in its own future and past.
Although it is not known whether CTCs are actually possible in our universe~\cite{Morris:1988gg, novikov1989analysis, Hawking1992, Visser:1995wormholesBook, Ori2005}, their mere logical possibility poses the challenge to understand what type of dynamics could be expected in their presence. 

The first systematic studies of the subject concentrated on space-time geometries where CTCs appear only in the {\em future\/} of some space-like surface~\cite{Friedman:1990ja, Echeverria:1991ko, Lossev1992} (Fig.~\ref{fig:ctcfuture}).
This makes it possible to set initial conditions on that surface, \ie, in the pre-CTCs era, and to look for the corresponding solutions to the equations of motion.  A prime case study is that of a billiard ball thrown in the direction of a wormhole: The initial position and velocity is chosen such that, if undisturbed, the ball comes out of the second mouth of the wormhole in the past and kicks its younger self off course, so the ball cannot reach the wormhole and kick itself. Since classical physics is clearly at odds with such `inconsistent' dynamics, the corresponding initial conditions would have to be `forbidden.' This, however, is itself at odds with the local nature of ordinary physical laws: What local mechanism prevents an experimenter from throwing the ball along the `impossible' trajectory?

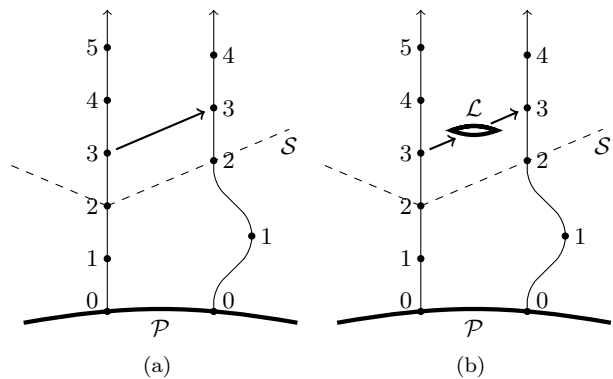
\begin{figure}
	\centering
	\subfloat[\label{fig:ctcfuture}]{
		\begin{tikzpicture}
			\node (X) {}; 
			\coordinate (C) at ($ (X.center) + (0.0,0.15) $);
			\coordinate (L) at ($ (X.center) - (1.8,0.0) $);
			\coordinate (R) at ($ (X.center) + (1.8,0.0) $);
			\coordinate (L0) at ($ (C) - (0.7,0) $);
			\draw[-,ultra thick] (L) to[out=10,in=170] node[midway,below] {$\mathcal{P}$} (R);
			\draw[->] (L0) -- ++(0,3.8);
			\draw[fill] (L0) circle (.3ex) node[left,yshift=0.15cm] {$0$};
			\foreach \x in {1,...,5}
				\draw[fill] ($ (L0) + (0,0.7*\x) $) circle (.3ex) node (M\x) { }node[left] {$\x$};
			\draw[->,rounded corners=5pt] ($ (C) + (0.7,0) $) -- ++(0,0.3) -- ++(0.5,0.5) -- ++(0,0.4) -- ++(-0.5,0.5) -- ++(0,2.1);
			\draw[fill] ($ (C) + (0.7,0) $) circle (.3ex) node (N0) {} node[right,yshift=0.15cm] {$0$};
			\draw[fill] ($ (C) + (1.2,1) $) circle (.3ex) node (N1) {} node[right] {$1$};
			\draw[fill] ($ (C) + (0.7,2) $) circle (.3ex) node (N2) {} node[right] {$2$};
			\draw[fill] ($ (C) + (0.7,2.7) $) circle (.3ex) node (N3) {} node[right] {$3$};
			\draw[fill] ($ (C) + (0.7,3.4) $) circle (.3ex) node (N4) {} node[right] {$4$};
			\draw[->,thick] (M3) -- (N3);
			\draw[-,dashed] ($ (L0) + (0,1.4) $) -- ++(23:2.6cm) node[below] {$\mathcal{S}$};
			\draw[-,dashed] ($ (L0) + (0,1.4) $) -- ++(180-23:1.4cm);
		\end{tikzpicture}
	}
	\subfloat[\label{fig:ctcnow}]{
		\begin{tikzpicture}
			\node (X) {};
			\coordinate (C) at ($ (X.center) + (0.0,0.15) $);
			\coordinate (L) at ($ (X.center) - (1.8,0.0) $);
			\coordinate (R) at ($ (X.center) + (1.8,0.0) $);
			\coordinate (L0) at ($ (C) - (0.7,0) $);
			\draw[-,ultra thick] (L) to[out=10,in=170] node[midway,below] {$\mathcal{P}$} (R);
			\draw[->] (L0) -- ++(0,3.8);
			\draw[fill] (L0) circle (.3ex) node[left,yshift=0.15cm] {$0$};
			\foreach \x in {1,...,5}
				\draw[fill] ($ (L0) + (0,0.7*\x) $) circle (.3ex) node (M\x) { }node[left] {$\x$};
			\draw[->,rounded corners=5pt] ($ (C) + (0.7,0) $) -- ++(0,0.3) -- ++(0.5,0.5) -- ++(0,0.4) -- ++(-0.5,0.5) -- ++(0,2.1);
			\draw[fill] ($ (C) + (0.7,0) $) circle (.3ex) node (N0) {} node[right,yshift=0.15cm] {$0$};
			\draw[fill] ($ (C) + (1.2,1) $) circle (.3ex) node (N1) {} node[right] {$1$};
			\draw[fill] ($ (C) + (0.7,2) $) circle (.3ex) node (N2) {} node[right] {$2$};
			\draw[fill] ($ (C) + (0.7,2.7) $) circle (.3ex) node (N3) {} node[right] {$3$};
			\draw[fill] ($ (C) + (0.7,3.4) $) circle (.3ex) node (N4) {} node[right] {$4$};

			\draw[-,dashed] ($ (L0) + (0,1.4) $) -- ++(23:2.6cm) node[below] {$\mathcal{S}$};
			\draw[-,dashed] ($ (L0) + (0,1.4) $) -- ++(180-23:1.4cm);
			\node (X2) at ($ (X.center) + (0,2.55) $) {};
			\node (Lab) at ($ (X2) + (0,0.3) $) {$\mathcal{L}$};
			\coordinate (L2) at ($ (X2.center) - (0.3,0.0) $);
			\coordinate (R2) at ($ (X2.center) + (0.3,0.0) $);
			\draw[-,ultra thick] (L2) to[out=+20,in=+160] (R2) to[out=-160,in=-20] (L2) to[out=20,in=160] (R2);
			\draw[->,thick] (M3) -- ++(23:0.50cm);
			\draw[->,thick] (M3)++(23:1cm) -- (N3);
		\end{tikzpicture}
	}
	\addfigref{\cite{Baumeler2018}}
	\caption{Wormhole space-time with closed time-like curves (CTCs)~\cite{Morris:1988gg, novikov1989analysis}. (a)~Events with equal proper times along the world lines of the two mouths of the wormhole are identified. Accelerating the right mouth produces time dilation, resulting in CTCs in the future of the surface~$\mathcal{S}$. An experimenter acting in the past of $\mathcal{S}$ should be able to prepare arbitrary initial states on a space-like surface~$\mathcal{P}$. (b)~An experimenter in a localised region~$\mathcal{L}$, which does not contain but is traversed by CTCs, should be able to perform arbitrary local operations.
}
	\label{fig:ctcs}
\end{figure}
\nofigref

The surprising result is that (possibly multiple) self-consistent solutions were found for all cases studied. The ball does not enter the wormhole undisturbed: It is kicked softly, comes out the wormhole at a slightly different angle than expected and gives its younger self just the right soft kick. Even including friction, exploding bombs, and the like, solutions for any considered initial condition were found~\cite{Novikov1992, Mikheeva1993}.

The existence of consistent solutions for every initial condition suggests a type of {\em `no new physics' principle}~\cite{Friedman:1990ja}: The presence of CTCs should not modify the local laws of physics, nor the range of possible initial states. 
It is then meaningful to ask whether the validity of the principle can be extended to the region where CTCs are {\em already\/} present~(Fig.~\ref{fig:ctcnow}). 
In such a case, there are no sufficiently regular space-like surfaces to set `global' initial conditions. Furthermore, there exist classes of space-times, aimed at simulating time machines, where geodesics are `reflected' by the CTC region, making it problematic to impose global initial conditions even in its past~\cite{Tippett2017}. However, the local nature of physical laws would imply that global features, such as the presence of CTCs, should not constrain the possible actions of an agent\footnote{We use the terms agent and party interchangeably. An agent (a party) can perform operations of her or his choice within a localized space-time region.} in any sufficiently {\em localised} region which itself does not contain CTCs.

Here we explore whether CTCs are compatible with such an extended requirement of `no new physics.'
Rather than considering a specific type of system (billiard balls, fields, {\it etc.}), we work within an abstract framework describing classical, deterministic dynamics that does not assume any particular causal structure.
It is formally a deterministic version of the formalism of `classical correlations without causal order'~\cite{Baumeler2016}, which in turn is the classical limit of the quantum `process matrix' formalism~\cite{Oreshkov:2012uh} (see also Refs.~\cite{baumeler14, araujo15, feixquantum2015, Oreshkov2015, Branciard2016, Giacomini2015, Baumann2016, oreshkov2016, abbott2016, araujo2016purification, Perinotti2016}). 
This latter formalism~\cite{Oreshkov:2012uh} can be used to study correlations among parties where the causal order among them is not specified {\it a priori}. The (often implicit) assumption of a causal order, instead, is replaced by the assumption that the probabilities of measurement outcomes are well defined. It has been shown that this formalism is more general than quantum theory: Correlations arise that cannot be explained with a predefined causal order of the parties. One can also consider the underlying state space to be classical (random variables) as opposed to quantum~\cite{Baumeler2016}. Strikingly, also in that special case such `non-causal' correlations arise.
Here, and as mentioned above, instead of dealing with quantum states or random variables, we restrict the model to {\em deterministic\/} dynamics; a limitation that is known to still allow for `non-causal' correlations~\cite{Baumeler2016}.

We prove that all classical, deterministic processes compatible with the free choice of local operations represent the evolution of classical systems via reversible dynamics in a suitable topology.
Surprisingly, this is not true in the quantum case, where certain processes are incompatible with reversible dynamics~\cite{araujo2016purification}; it is neither true if the underlying states are random variables~\cite{Baumeler2015}.
Yet, for deterministic dynamics, every process can be turned into a reversible one, as is shown later.

We further provide a complete characterisation of up to tripartite processes, including continuous-variable systems, identifying a broad class of processes that can only be realised in the presence of CTCs. Our results show that {\em CTCs in general are logically compatible with classical, deterministic dynamics}, where agents are {\em free\/} to perform any classical operation in local regions and, outside such regions, systems evolve according to classical, {\em reversible\/} dynamics.

Our approach departs from previous models of CTCs directly based on the quantum formalism~\cite{Deutsch:1991jo, Politzer1994, Pegg:2001wa, Bennett, Bacon2004, Svetlichny:2009ve,Svetlichny:2011gq, Lloyd2011exp, Lloyd:2011ir, Ralph2010, Ralph:2012cd, Wallman2012, Allen2014}.
	In those models, if the underlying systems that travel on the CTCs are classical, then the `no new physics' principle and the assumption of free choice {\em cannot\/} be uphold simultaneously.
	Rather, it was believed that one needs to invoke quantum mechanics in order to restore these desired features.
	In this article, we finally discuss the extension of the presented formalism to quantum CTCs and discuss it in the light of the above mentioned models.
	A detailed comparison among these models will be provided in a forthcoming paper~\cite{Baumeler2016b}.

\section{Approach}
It is customary to understand physical laws as rules for {\em predicting\/} future events based on initial conditions. This is often seen as an essential aspect for extracting testable predictions: One can always conceive, at least in principle, a controlled experiment where the initial conditions are set independently of any other relevant aspect of the experiment, and the final state is measured.	 
CTCs undermine this view: In a general scenario, where events can be in their own future and global space-like surfaces are not available, it is unclear what independent variables an experimenter could try to control. Solutions to the dynamics have to be determined ``all at once,'' without a clear place for external interventions.
	
Our approach to a {\em predictive framework\/} for not globally hyperbolic space-times,\footnote{A space-time is \emph{globally hyperbolic} if and only if it contains a Cauchy surface, \ie, dynamics on it has a well defined initial-condition problem \cite{Geroch1970}.} possibly containing CTCs, is to move to a more general notion of intervention, partly inspired by the methodology of classical causal models \cite{Pearlbook, woodward2003making}. The key assumption is that physical laws retain their modular character, namely it is possible to separate the physical properties relevant for the dynamics of the system of interest from those responsible for the ``intervention.'' For example, when studying the trajectory of a billiard ball, one can typically ignore the particular mechanism setting the ball in motion; or when estimating the radiation produced by moving charges, one can abstract away the force causing the charge's movement. In other words, we wish to retain a notion of ``freedom of choice'' in manipulating the relevant variables.\footnote{This assumption does not imply any metaphysical commitment regarding human ``free will.''}

To be more precise, we will identify interventions with localised operations in CTC-free regions of space-time. Following the above assumption, we assume that it is possible---at least in principle---to engineer any local operation on the system of interest without significantly affecting the dynamics outside the intervention regions. The role of the dynamics is then to predict how a system will respond to certain operations. In other words, given the specification of the operations in a set of space-time regions where intervention takes place, we should be able to calculate the state of the system observed in any other space-time region. We will call ``process'' the function providing such a prediction.

Note that a ``process,'' in general, includes both information about the dynamics as well as additional boundary conditions not set by intervention. To clarify this point, consider an ordinary, CTC-free, scenario where an intervention sets the value of some variable ({\it e.g.}, a field) on a small portion of a space-like surface and we want to estimate the resulting field somewhere in its causal future.  In general, this requires fixing additional boundary conditions, for example on an extension of the initial region to a Cauchy surface. The ``process'' is the function mapping the field values in the small region to the observed final field, once the remaining boundary conditions are fixed. In a limiting case, where no intervention is made, the ``process'' is simply a specification of an admissible boundary condition in the region where an observation is made.

According to the above scheme, a particular dynamical law generates a \emph{class} of processes, describing how interventions in arbitrary regions can influence observations in other regions, as well as all possible boundary conditions consistent with the laws. In the case of not globally hyperbolic space-times, boundary conditions might be subject to non-trivial constraints. Our working assumption is that such constraints should not affect the type of operations available in small regions that do not contain CTCs. This might at first seem at odds with `grandfather paradox'-type arguments: Denote by $x$ the value a physical property takes on the future boundary of our region, $a$ the value on the past, and let the `CTC dynamics' be the identity, $a=x$. This is incompatible with any local operation $x=f(a)$ other than identity. 

The problem with the argument above is that it simply \emph{assumes} that the identity backwards in time is a possible solution of the dynamics, based on the intuition that such evolution would be possible if $a$ were in the future of $x$, without CTCs. The studies mentioned in the introduction suggest that such an assumption is typically incorrect: The system's evolution typically finds a way to `adjust itself,' preserving the consistency of `free interventions.' The upshot is that we should not expect all conceivable functions to appear as possible solutions of the dynamics. Our approach is to take consistency with local operations as a starting point and explore the consequences: Possible dynamics and their respective predictions are calculated based on the ``freedom of choice'' of local operations. Whether this assumption is valid would depend on the particular geometry and dynamical equations. The main result is that it is logically possible, at least in principle, to have a predictably sound theory, compatible with local interventions, where the presence of CTCs can be made manifest through appropriate experiments.

\section{The formalism}
The core assumption of our model is that any classical operation that is possible in an ordinary space-time should also be possible in the presence of CTCs, as long as the operation takes place in a localised region of space-time that does not contain CTCs.
This states that localised regions are {\em ignorant\/} of CTCs. Let us elaborate more on this core assumption.

We consider $N$ non-overlapping space-time regions (henceforth {\em local regions}) which, individually, cannot be distinguished from regions in ordinary, globally hyperbolic space-time. These are the regions in which `interventions' can be made; \ie, evolution inside the regions is assumed to be arbitrary, while outside the regions it is fixed by some given dynamics and boundary conditions. We impose no restriction on the space-time in which the regions are embedded, except that it is a Lorentzian manifold fixed independently of any dynamical degree of freedom of interest. Details regarding the causal structure of space-time will not play a prominent role in our analysis, but we point to Ref.~\cite{Minguzzi2019} for a recent review.

To simplify the analysis, we restrict to compact, simply connected regions that have only space-like boundaries, which we decompose into a past boundary and a future boundary (these local regions are therefore space-time local), see Fig.~\ref{fig:localregion}.
Furthermore, and for the same purpose, we assume that, for each local region, any time-like curve that enters through the past (future) boundary exits through the future (past) boundary, and that the region contains no CTCs. These assumptions ensure that we can treat the local regions as `closed laboratories'~\cite{Oreshkov:2012uh}---where a system can only enter once through the past boundary and exit once through the future boundary, without exchanging information with the exterior in between. Under such conditions, the local operations can be simply identified with transformations from an input to an output space. The assumption of space-like boundaries has the further role of preventing, in a globally hyperbolic space-time, that a region intersects both the future and past light cone of another region. In other words, causal order among regions would form a partial-order relation. Therefore, any departure from partial order reveals a not globally hyperbolic space-time and some non-trivial form of time travel.

As we are interested in classical systems, we can assign classical state spaces $\mathcal{I}_R$ (input) and $\mathcal{O}_R$ (output) respectively to the past and future boundaries of a local region $R$. States will be denoted as $i_R\in\mathcal{I}_R$, $o_R\in\mathcal{O}_R$. {For example, in a field theory, a state would be a function on the corresponding boundary surface and the state spaces would be appropriate spaces of functions.}

A {\em deterministic local operation\/} in the local region is described by a function $f_R$ from input to output space~(Fig.~\ref{fig:localregion}). We denote by $\mathcal{D}_R:=\left\{f_R:\mathcal{I}_R \rightarrow \mathcal{O}_R\right\}$ the set of all possible operations in region $R$.
\begin{figure}
	\centering
	\subfloat[\label{fig:localregion}]{
		\begin{tikzpicture}
				\node (X) {};
				\coordinate (L) at ($ (X.center) - (1.0,0.0) $);
				\coordinate (R) at ($ (X.center) + (1.0,0.0) $);
				\path (L) to[out=+20,in=+160] node[midway] (A) {} (R);
				\path (L) to[out=-20,in=-160] node[midway] (B) {} (R);
				\draw[-latex,double] (B.center) -- node[midway,right] {$f_R$} (A.center);
				\draw[-,ultra thick] (L) to[out=+20,in=+160] node[midway,above] {$o_R$} (R)
							 to[out=-160,in=-20] node[midway,below] {$i_R$} (L) to[out=20,in=160] (R);
		\end{tikzpicture}
	}
	\quad
	\subfloat[\label{fig:source}]{
		\begin{tikzpicture}
				\node (X) {};
				\coordinate (L) at ($ (X.center) - (1.0,0.0) $);
				\coordinate (R) at ($ (X.center) + (1.0,0.0) $);
				\draw[-,ultra thick] (L) to[out=+20,in=+160] node[midway,above] {$o_R$} (R);
		\end{tikzpicture}
	}
	\quad
	\subfloat[\label{fig:sink}]{
		\begin{tikzpicture}
				\node (X) {};
				\coordinate (L) at ($ (X.center) - (1.0,0.0) $);
				\coordinate (R) at ($ (X.center) + (1.0,0.0) $);
				\draw[-,ultra thick] (L) to[out=-20,in=-160] node[midway,below] {$i_R$} (R);
		\end{tikzpicture}
	}
	\caption{(a) A deterministic local operation~$f_R$ maps the classical input state~$i_R$ from the past boundary to the classical output state~$o_R$ at the future boundary of a local region~$R$. (b)~An 'output only' region (source). (c) An 'input only' region (sink).}
	\label{fig:localoperation}
\end{figure}
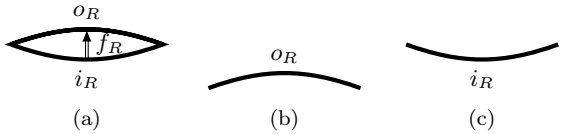
We drop the index to refer to collections of objects for all regions, as in \mbox{$i\equiv \{i_1,\dots,i_N\}$}, $\mathcal{I}\equiv\mathcal{I}_1\times\dots\times \mathcal{I}_N$, \mbox{$\mathcal{D}\equiv\mathcal{D}_1\times\dots\times \mathcal{D}_N$},~{\it etc.}\footnote{Note that $\mathcal{D}$ is not the set of {\em all\/} functions $\mathcal{I} \rightarrow \mathcal{O}$, but rather of those of the form $f(i)=\left\{f_1(i_1),\dots,f_N(i_N)\right\}$.} Local operations are not required to be reversible, \ie, the local functions $f_R$ need not be invertible. This corresponds to the assumption that the local experimenters and devices have the ability to erase information by accessing some reservoir, not included in the description of the physical degrees of freedom of interest. Furthermore, input and output state spaces need not be isomorphic, as degrees of freedom may be added or removed during the operation. We will also consider the special case in which either input or output state space is the empty set. An `output only' region---called a {\em source}---can be identified with a space-like region on which an agent (acting somewhere in its past) can prepare an arbitrary state~(Fig.~\ref{fig:source}), while an `input only' region---called a {\em sink}---can be identified with a space-like region where an agent (somewhere in its future) can only observe the state~(Fig.~\ref{fig:sink}). Ordinary dynamics is concerned with the evolution from a source (state preparation on a space-like surface) to a sink (state observation on a space-like surface).

We want to define a generalised type of dynamics for an arbitrary number of regions---in which arbitrary classical operations can be performed---possibly embedded in a space-time with CTCs. The basic requirement of such a model is that it must be able to {\em predict\/} the state observed on the past boundary of each local region,\footnote{The state on the future boundary of each local region is obtained by applying the local operation on the state observed on the past boundary of that local region.} which in general can depend on all local operations. (In a CTC-free space-time, the input state on a space-like region would only depend on operations in its past light-cone). For a deterministic model, such a dependence is encoded in a function $\omega\equiv\{\omega_1,\dots,\omega_N\}:\mathcal{D}\rightarrow \mathcal{I}$ that determines the state on the past boundary of each region as a function of all local operations. 

We define consistency with arbitrary choices of local operations in the following way: Given a set of operations $f\in \mathcal{D}$, let $i_R=\omega_R(f)$ be the input state of region $R$. This is transformed into the output $o_R=f_R(i_R)$ by the local operation in that region. However, there are several different functions that yield the same output $o_R$ given the input $i_R$. Since the experimenter is free to choose any operation, and it should not matter when that choice is made (in particular, it can be made {\em after\/} she already knows the input $i_R$), the external dynamics should not distinguish between all such operations. In particular, let us define the constant operation $C_o(i)=o$ yielding outcome $o$ for any input $i$. Consistency with arbitrary local operations is then formulated as
\begin{equation}
\forall f\in \mathcal{D},\quad \omega(f) = \omega\left(C_{f(\omega(f))}\right)\,.
\label{consistency}
\end{equation} 
In words, the input states $i=\omega(f)$ produced by the process given local operations $f$ should be the same as the states produced when all parties perform the constant operations that yield the same output states as the original operations $f$ applied to $\omega(f)$.

We can simplify the representation of a process $\omega$ by noticing that it defines a unique function $w:\mathcal{O}\rightarrow \mathcal{I}$, $w(o):=\omega(C_o)$, which we will call \emph{process function}. This provides an alternative view for defining a process as a function that maps the states in the future boundary of the local regions to states in the past boundaries. This is indeed what we expect for local dynamical equations relating states at different points in space-time.

By writing $\omega(f)=i$, we see that condition~\eqref{consistency} implies the following condition for~$w$:
\begin{align}
	\forall f\in \mathcal{D},\, \exists i\in \mathcal{I} \quad \text{such that}\quad w\circ f (i)  = i\,.
	\label{fixedpoint}
\end{align}
In other words, if $w$ is a process function, then $w\circ f$ has a fixed point for every local operation $f$.

Condition~\eqref{fixedpoint} is in fact both necessary and sufficient: A function $w$ satisfying~\eqref{fixedpoint} {\em uniquely\/} defines a process. This is because of the uniqueness of the fixed points:
\begin{theorem}[Unique fixed points]
	\label{fixedthm}
Given a function $w:\mathcal{O}\rightarrow \mathcal{I}$ that satisfies condition~\eqref{fixedpoint}, the fixed point of $w\circ f$ is unique for every set of local operations \mbox{$f=\{f_1,\dots,f_N\}\in \mathcal{D}$}.
\end{theorem}
We prove this theorem in the Appendix. As opposed to the proof of the analogous theorem in the probabilistic version of the formalism~\cite{Baumeler2015}, our proof also holds for continuous and not only discrete variables~\cite{askus}.

Because of Theorem~\ref{fixedthm}, every function $w:\mathcal{O}\rightarrow \mathcal{I}$ that satisfies condition~\eqref{fixedpoint} defines a unique function $\omega:\mathcal{D}\rightarrow \mathcal{I}$, with $\omega(f)$ equal to the unique fixed point of $w\circ f$. It is furthermore easy to see that condition~\eqref{fixedpoint} implies the consistency condition~\eqref{consistency}. Therefore, we can identify a process with its process function $w$. The interpretation is that dynamics in the presence of CTCs is described by a function that maps the states on the future boundaries of all regions to states on the past boundaries of each region. Condition~\eqref{fixedpoint} imposes that such a dynamics is compatible with arbitrary operations in each region; Theorem~\ref{fixedthm} further guarantees that specifying the operations performed in each region is sufficient to predict a {\em unique\/} state on each of the past boundaries.

\section{Reversibility}
Reversible classical dynamics is associated with invertible functions, such that the role of `preparation' and `measurement' can be swapped. Not all process functions are invertible; for example, the process function for a single `sink' region (with trivial output) reduces to the specification of a state on that region and it is clearly not invertible. However, such a process function can be extended to a reversible one by introducing a `source' region (with trivial input), in the past of the sink, so that the state on the sink can now be calculated as a function of the state prepared by the source, and this function can be invertible.

Crucially, we can prove that {\em every\/} process function can be extended to an invertible one, as expressed by the following theorem, proved in the Appendix.

\begin{theorem}[Reversibility]
	\label{thm:go}
	For every function \mbox{$w:\mathcal{O}\rightarrow \mathcal{I}$} that satisfies condition~\eqref{fixedpoint}, there exists an invertible function $w':\mathcal{O}\times\mathcal{O}_\source \rightarrow \mathcal{I}\times\mathcal{I}_\sink$, where $\mathcal{O}_\source$ is the output space of a region with trivial input (the {\em `source'}) and $\mathcal{I}_\sink$ is the input space of a region with trivial output (the {\em `sink'}), such that $w'$ satisfies condition~\eqref{fixedpoint} and there exists $\tilde{o}_\source\in \mathcal{O}_\source$ such that $w'\left(o,\tilde{o}_\source\right)=\left\{w(o),g_\sink(o)\right\}$ for some invertible function $g_\sink$.
\end{theorem}

This theorem shows that all process functions can be interpreted in terms of reversible dynamics: The source describes a space-like region `in the past' of all other regions, while the sink is a space-like region `in the future.'
The process determines the state of the sink as well as the states on the past boundaries of all regions as a function of the states on the future boundaries of all regions and of the source. Because it is reversible, the process can be read in the opposite direction: Given the states on the sink and on the past boundaries, it allows calculating the state on all future boundaries and of the source. The time-reversed process is then compatible with arbitrary reversed local operations that map local outputs to local inputs. In Ref.~\cite{Baumeler2016b}, it is further proven that the presence of a source and a sink is necessary in order to define a reversible process.

\section{Characterization of process functions}

The causal relations encoded by process functions can be better understood in terms of \emph{signalling}. In general, given a function $h:\mathcal{A}\times\mathcal{B}\rightarrow\mathcal{C}$, we say \emph{there is no signalling} from $\mathcal{A}$ to $\mathcal{C}$ if, for every $b\in\mathcal{B}$, 
\begin{equation}
h(a,b) = h(\tilde a,b) \quad \forall\; a, \tilde a \in \mathcal{A}.
\label{nosignalling}
\end{equation}
We say there is signalling if the opposite is true, \ie, $h(a,b) \neq h(\tilde a,b)$ for some $a$, $\tilde a$, $b$. The terminology applies to process functions in the obvious way: A region $S$ does not signal to a region $R$ if the $R$-th component of the process function does not depend on the output of $S$.

Simple examples of process functions are causally ordered ones, namely those compatible with CTC-free dynamics, for which signalling is only possible in one direction.\footnote{Formally, a set of regions is causally ordered if their causal relations in space-time form a partial order: Any region $R$ is either in the causal past, causal future or space-like from any other region $S$. A process function is compatible with such a structure if signalling is only possible from a region to its causal future. We call such a process function causally ordered.} For example, for regions $R,S,T,\dots$ a process function $w\equiv \{w_R,w_S,w_T,\dots\}$ compatible with the causal order $R\prec S \prec T\prec\dots$ is given by \mbox{$w_R(o_R,o_S)=\bar{i}_R$} (constant), \mbox{$w_S(o_R,o_S) = w_S(o_R)$}, \mbox{$w_T(o_R,o_S,o_T) = w_T(o_R,o_S)$},~{\it etc.}, where the state on each past boundary is independent of future states.\footnote{For the sake of presentation, equations like~$w_S(o_R,o_S) = w_S(o_R)$ express the independence of the function value~$w_S(o_R,o_S)$ from~$o_S$, \ie,~$\forall o_R,o_S,\tilde o_S: w_S(o_R,o_S) = w_S(o_R,\tilde o_S)$.} It is easy to see that condition~\eqref{fixedpoint} is satisfied in such cases, \ie, a fixed point exists for every choice of local operations (it is given by $i_R=\bar{i}_R$, $i_S=w_S\circ f_R(\bar{i}_R)$, and so on). The interesting question is whether more general processes are possible, once CTCs are allowed.
To answer this question, we will first give a complete characterisation of all process functions for up to three regions. The detailed proofs can be found in the Appendix.

For a single local region, a process function has to be a constant: $w(o)=\bar{i}$ $\forall o$.
Thus, an observer acting in a localised region cannot send information back to herself; her observations are fully compatible with her region being embedded in a CTC-free space-time.
A direct consequence is that, for an arbitrary number of regions, the input of each region $R$ cannot depend on that region's output:
$	w_R(o)=w_R(o_{\setminus R})$, 
where $o_{\setminus R}$ is the set of outputs of all regions except $R$.

Bipartite process functions are characterized by the following conditions:
\begin{enumerate}
	\item[(i)] $w_R(o_R,o_S)=w_R(o_S)\,$,
	\item[(ii)] $w_S(o_R,o_S)=w_S(o_R)\,$,
	\item[(iii)] at least one of~$w_R(o_S)$ or~$w_S(o_R)$ is constant.
\end{enumerate}
In other words, deterministic process functions can only allow one-way signaling. Again, two observers in distinct localised regions would not be able to verify the presence of CTCs outside their regions. (Remarkably, this is not true for the quantum version of the framework~\cite{Oreshkov:2012uh}.)

Consider now three regions $R$, $S$, $T$. For simplicity, we denote input and output variables as $a\in \mathcal{A}$, $b\in \mathcal{B}$, $c\in \mathcal{C}$ and $x\in \mathcal{X}$, $y\in \mathcal{Y}$, $z\in \mathcal{Z}$, respectively. A process function has then three component functions: $a=w_R(y,z)$, $b=w_S(x,z)$, $c=w_T(x,y)$ (where we used the fact that the input of each region cannot depend on its own output, as seen above). We give a simple characterization of process functions as functions where the output variable of one region `switches' the direction of causal influence between the two other parties.
\begin{theorem}[Tripartite process function]\label{3charact}
	Three functions $w_R:\mathcal{Y}\times \mathcal{Z}\rightarrow \mathcal{A}$, $w_S:\mathcal{X}\times \mathcal{Z}\rightarrow \mathcal{B}$, $w_T:\mathcal{X}\times \mathcal{Y}\rightarrow \mathcal{C}$ define a process function if and only if each of the three ``reduced functions'' defined as
	\begin{align}
	w^z(x,y):=&\left\{w_R(y,z), w_S(x,z)\right\},\\
	w^x(y,z):=&\left\{w_S(x,z), w_T(x, y)\right\},\\
	w^y(x,z):=&\left\{w_R(y,z), w_T(x,y)\right\}
	\end{align}
is a bipartite process function for every $z\in \mathcal{Z}$, $x\in \mathcal{X}$, $y\in \mathcal{Y}$ respectively.
\end{theorem}
Recall that a bipartite process function is at most one-way signaling. Theorem \ref{3charact} thus says that $w$ is a tripartite process function if and only if, for every fixed value for the outcome of one of the regions, only one-way signaling is possible between the other two. It is an open question whether a similar condition characterises arbitrary multipartite processes.

\section{Examples}

Given the above characterisation, it is simple to find process functions that cannot arise in ordinary, causal space time: It is sufficient that each party can signal non-trivially to the other two, while satisfying the condition of Theorem \ref{3charact}. We present a continuous-variable example, based on a similar process for `bits,' first found by Ara\'ujo and Feix and published in Ref.~\cite{Baumeler2016}. Consider a tripartite scenario as above, where $x,y,z,a,b,c \in \mathbb{R}$. We define $w:\ \mathbb{R}^3 \rightarrow \mathbb{R}^3$ as $(x,y,z) \mapsto (a,b,c)$, with
\begin{equation}
\label{cont}
	a= \Theta(-y) \Theta(z)\,, \quad
	b = \Theta(-z) \Theta(x)\,, \quad
	c= \Theta(-x) \Theta(y) 
	\,,
\end{equation}
where $\Theta(t)=1$ for $t>0$, $\Theta(t)=0$ for $t\leq 0$. In this process, the sign of the output of each region determines the direction of signaling between the other two. For example, for $y\leq 0$ we have $a=\Theta(z)$ ($T$ can signal to~$R$) but $c=0$ ($R$ cannot signal to $T$), while for $y>0$ the opposite direction of signaling holds (and similarly for the other pairs of regions). 

By Theorem~\ref{thm:go}, we can extend~$w$ to a {\em reversible\/} process function~$w'$. To this end, we introduce source and sink spaces, both isomorphic to $\mathbb{R}^3$, with variables $e_0, e_1, e_2$ and $s_0, s_1, s_2$, respectively. The extended process function $w':\ \mathbb{R}^6 \rightarrow \mathbb{R}^6$ is given by
\begin{align}
	\label{eq:ctc}
	&a= \Theta(-y) \Theta(z) + e_0\,,\qquad s_0= x\,, \\ \nonumber
	&b= \Theta(-z) \Theta(x) + e_1\,,\qquad s_1= y\,,\\ \nonumber
	&c= \Theta(-x) \Theta(y) +e_2\,, \qquad s_2= z\,.
\end{align}
This process allows three observers in regions $R,S,T$ to perform arbitrary deterministic operations on the system they receive from the respective past boundary, sending the result out the respective future boundary. The outgoing systems then enter the CTC region and undergo some reversible transformation, interacting with each other and with the output of the source~$\source$, eventually determining the state in the past of each region and of the sink~$\sink$~(Fig.~\ref{fig:example}).
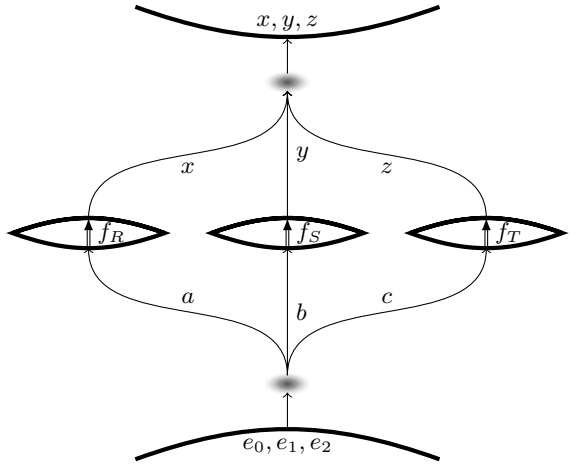
\begin{figure}
        \centering
        \begin{tikzpicture}
		\node (R) {};
		\coordinate (RL) at ($ (R.center) - (1.0,0.0) $);
		\coordinate (RR) at ($ (R.center) + (1.0,0.0) $);
		\path (RL) to[out=+20,in=+160] node[midway] (RA) {} (RR);
		\path (RL) to[out=-20,in=-160] node[midway] (RB) {} (RR);
		\draw[-latex,double] (RB.center) -- node[midway,right] {$f_R$} (RA.center);
		\draw[-,ultra thick] (RL) to[out=20,  in=160] (RR)
					  to[out=-160,in=-20] (RL)
					  to[out=20,  in=160] (RR);
		\node[right=2.5cm of R.center] (S) {};
		\coordinate (SL) at ($ (S.center) - (1.0,0.0) $);
		\coordinate (SR) at ($ (S.center) + (1.0,0.0) $);
		\path (SL) to[out=+20,in=+160] node[midway] (SA) {} (SR);
		\path (SL) to[out=-20,in=-160] node[midway] (SB) {} (SR);
		\draw[-latex,double] (SB.center) -- node[midway,right] {$f_S$} (SA.center);
		\draw[-,ultra thick] (SL) to[out=20,  in=160] (SR)
					  to[out=-160,in=-20] (SL)
					  to[out=20,  in=160] (SR);
		\node[right=2.5cm of S.center] (T) {};
		\coordinate (TL) at ($ (T.center) - (1.0,0.0) $);
		\coordinate (TR) at ($ (T.center) + (1.0,0.0) $);
		\path (TL) to[out=+20,in=+160] node[midway] (TA) {} (TR);
		\path (TL) to[out=-20,in=-160] node[midway] (TB) {} (TR);
		\draw[-latex,double] (TB.center) -- node[midway,right] {$f_T$} (TA.center);
		\draw[-,ultra thick] (TL) to[out=20,  in=160] (TR)
					  to[out=-160,in=-20] (TL)
					  to[out=20,  in=160] (TR);
		\mouth{(S)++(0.0,+2)}{F}
		\mouth{(S)++(0.0,-2)}{P}
		\coordinate (sink)   at ($ (F.center) + (0,1.0) $);
		\coordinate (sinkL)  at ($ (sink.center) - (2.0,0.0) $);
		\coordinate (sinkR)  at ($ (sink.center) + (2.0,0.0) $);
		\draw[-,ultra thick] (sinkL) to[out=-20,in=-160] (sinkR);
		\coordinate (source)   at ($ (P.center) - (0,1.0) $);
		\coordinate (sourceL)  at ($ (source.center) - (2.0,0.0) $);
		\coordinate (sourceR)  at ($ (source.center) + (2.0,0.0) $);
		\draw[-,ultra thick] (sourceL) to[out=+20,in=+160] (sourceR);
		\draw[->] (RA.center) to[out=90,in=270] node[midway,below] {$x$} (F.south);
		\draw[->] (SA.center) to[out=90,in=270] node[midway,right] {$y$} (F.south);
		\draw[->] (TA.center) to[out=90,in=270] node[midway,below] {$z$} (F.south);
		\draw[->] (P.north) to[out=90,in=270] node[midway,above] {$a$} (RB.center);
		\draw[->] (P.north) to[out=90,in=270] node[midway,right] {$b$} (SB.center);
		\draw[->] (P.north) to[out=90,in=270] node[midway,above] {$c$} (TB.center);
		\draw[->] (F.north) -- ++(0,+0.46) node[above] {$x,y,z$};
		\draw[<-] (P.south) -- ++(0,-0.5) node[below] {$e_0,e_1,e_2$};
		\draw[->,dashed] (F.330) to[out=330,in=30] (P.30);
	\end{tikzpicture}
	\caption{The output of three local regions fall into a CTC where they undergo a joint interaction with the state prepared by the source. The CTC outputs the input states to the three local regions and the sink.}
	\label{fig:example}
\end{figure}
Crucially, the input state of each region depends non-trivially on the output state of the other two, thus each observer can communicate to every other. Thus, we have a situation where three observers can experimentally verify to be each {\em both in the past and in the future\/} of each other, they can perform {\em arbitrary\/} local operations, and {\em no contradiction\/} ever emerges.

\section{Quantum closed time-like curves}
The above framework of classical, reversible dynamics can be extended to quantum systems. It is then interesting to compare the resulting model with existing quantum models for CTCs. We briefly present here the main results, and refer to Ref.~\cite{Baumeler2016b} for a detailed analysis. 

A classical system can be `quantised' by associating to each state a distinct orthogonal state in a Hilbert space, with quantum superpositions represented by linear combinations. Thus, in the quantum version of the formalism, the boundary of each region is associated with a Hilbert space. A classical, reversible process defines a permutation of basis elements and can be extended by linearity to the entire Hilbert space, defining a unitary map from the future to the past boundaries. It is not {\it a priori\/} guaranteed that such a unitary defines a valid {\em quantum process}: Observers in the local regions should now be able to perform arbitrary {\em quantum\/} operations. The resulting constraints on quantum processes can be conveniently formalised using the {\em process matrix\/} formalism of Ref.~\cite{Oreshkov:2012uh}.
Using the characterisation of tripartite quantum processes of Ref.~\cite{araujo15}, it is proven in Ref.~\cite{Baumeler2016b} that the quantisation of a finite-dimensional version of~Eq.~\eqref{cont} indeed defines a {\em valid\/} unitary quantum process.

The two most studied models of quantum systems in the presence of CTCs are the so-called post-selected CTC model (P-CTC)~\cite{Politzer1994, Pegg:2001wa,Bennett,Svetlichny:2009ve,Svetlichny:2011gq,Lloyd:2011ir} and the Deutsch model~\mbox{(D-CTC)}~\cite{Deutsch:1991jo, Bacon2004, Ralph2010, Ralph:2012cd, Wallman2012}.
Both models assume that CTCs are only present in a limited portion of space-time. At some time before the CTCs, a chronology-respecting (CR) system is prepared. Then, the CR system interacts with a chronology-violating (CV) one, which travels along a CTC. The models prescribe how to calculate the state of the CR system obtained after interaction with the CV one. Within such frameworks, we can model the multi-region scenarios considered here by introducing a CR and a CV systems per region, and interpreting the interaction between each pair as our local operation in the corresponding local region. The CV systems then interact according to the unitary process and are later sent back in time, with the backward evolution described according to the specific model. We can then compare the evolution of the CV system predicted by each model.

As it turns out, the P-CTC model gives the same predictions as ours for any {\em valid\/} unitary process. The crucial difference is that the P-CTC model allows the CV system to evolve according to {\em arbitrary\/} unitaries, generically resulting in a {\em non-linear\/} evolution for the CR system and in a restriction on the local operations that can be performed. By contrast, our model imposes additional constraints, which effectively enforce the CR system to evolve {\em linearly}. 
The D-CTC model, on the other hand, allows arbitrary operations to be performed locally. However, it predicts {\em non-linear\/} evolution of the CR system, even when the CV system evolves according to a process subject to the constraints introduced here~\cite{Baumeler2016c}.
\section{Conclusions}

We developed a framework for deterministic dynamics in the presence of CTCs. The framework extends the ordinary concept of time evolution---where a future state is calculated as a function of a past one---to the more general scenario of a number of space-time regions, where the state in the past of each region is calculated as a function of the state in the future of all regions. Requiring that {\em arbitrary\/} operations must be possible in each region imposes strong constraints on the allowed dynamics. Our main result is that it is possible to have {\em reversible\/} dynamics, compatible with {\em arbitrary\/} local operations, where the state observed in each region {\em depends non-trivially on the states prepared in all other regions}. Because such a functional relation is reversible, it can always be realised by some physical system subject to local dynamical laws, {\it e.g.}, in terms of a system of bouncing billiard balls~\cite{Fredkin1982}. 

The main message of our result is that CTCs are not necessarily in conflict with local physics, nor with the `freedom of choice' associated with the possibility of performing arbitrary operations.
The latter furthermore implies that the choice of the operations together with the CTC {\em uniquely\/} determines the states on the past boundaries.
Importantly, and contrarily to several previous approaches~\cite{Deutsch:1991jo, Politzer1994, Pegg:2001wa, Bennett, Bacon2004, Svetlichny:2009ve,Svetlichny:2011gq, Lloyd:2011ir,  Ralph2010, Ralph:2012cd, Wallman2012, Allen2014}, quantum mechanics need not be invoked to `solve' paradoxes of classical time travelling. A quantum version of the framework can be developed within the so-called ``process matrix'' formalism~\cite{Oreshkov:2012uh}, where it would be natural (in analogy to classical determinism and reversibility) to impose unitarity of the process~\cite{araujo2016purification}. The precise connection between classical and quantum frameworks remains however an open question---we provided a brief discussion nevertheless---for example, it is unclear weather all classical, deterministic processes can be ``quantised'' to give a valid unitary quantum process.

	{\bf Acknowledgments.}
	We thank M.~Ara\'ujo, V.~Baumann, J.~Bowles, \v C.~Brukner, J.~Degorre, P.~Erker, A.~Feix, C.~Giarmatzi, A.~Hansen, A.~Montina, M.~Navascu\'es, O.~Oreshkov for helpful discussions, and anonymous reviewers for their comments.
	The present work was supported in part by the Swiss National Science Foundation (SNF), the National Centre of Competence in Research `Quantum Science and Technology' (QSIT), and the Templeton World Charity Foundation (TWCF 0064/AB38). \"A.B.\ acknowledges the Swiss National Science Foundation (SNSF) under grant 175860, the Erwin Schr\"odinger Center for Quantum Science \& Technology (ESQ), and the Austrian Science Fund (FWF): ZK 3. F.C.\ \& M.Z.\ acknowledge support through  Australian Research Council Discovery Early Career Researcher Awards (DE170100712 \& DE180101443). This publication was made possible through the support of a grant from the John Templeton Foundation. 
The opinions expressed in this publication are those of the authors and do not necessarily reflect the views of the John Templeton Foundation. We acknowledge the traditional owners of the land on which the University of Queensland is situated, the Turrbal and Jagera people.

\begin{appendix}

\section{Properties of process functions}
Here we derive a set of properties of process functions, which will be needed in later proofs. For convenience, we will use the term \emph{process function} to denote any function $w$ that satisfies condition~(\ref*{fixedpoint}) in the main text, namely that a fixed point of $w\circ f$ exists for each $f \in \mathcal{D}$, even though the equivalence between this condition and the main-text definition of process function, Eq.~(\ref*{consistency}), is due to Theorem~\ref*{fixedthm}, which is proved in the next section.

Let us first fix some notation. As in the main text, an object without index refers to a collection of objects: $\mathcal{I}= \mathcal{I}_1\times\dots\times\mathcal{I}_N$, {\it etc.} We will also use the notation $\mathcal{I}_{\setminus R}= \mathcal{I}_1\times\dots\times\mathcal{I}_{R-1}\times \mathcal{I}_{R+1}\times\dots\times\mathcal{I}_N$,  $i_{\setminus R}=\left\{i_1,\dots i_{R-1},i_{R+1},\dots,i_N\right\}$, {\it etc.}, to denote collections with the component $R$ removed. Appropriate reordering will be understood when joining variables, for example in expressions as $i=i_R\cup i_{\setminus R}$, $f(i)=f(i_R,i_{\setminus R})$, and so on.
Moreover, expressions like~$w_R(o)=w_R(o_{\setminus R})$ denote that region~$R$ {\em cannot\/} signal to itself, {\it i.e.},~$\forall o_{\setminus R},o_R,\tilde o_R: w_R(o_{\setminus R},o_R)=w_R(o_{\setminus R},\tilde o_R)$.

The first property we need is a necessary condition for process functions:
\begin{lemma}\label{constant}
For a process function $w$, each component $w_R: \mathcal{O}\rightarrow \mathcal{I}_R$ must be constant over $\mathcal{O}_R$: $w_R(o)=w_R(o_{\setminus R})$ for $R=1,\dots,N$.
\end{lemma}
\begin{proof}
For some set of local operations $f=\left\{f_R:\mathcal{I}_R\rightarrow \mathcal{O}_R\right\}_{R=1}^N$ and some fixed $\bar{i}_{\setminus R}$, let us define $h_R:\mathcal{I}_R\rightarrow \mathcal{I}_R$ as $h_R(i_R):=w_R\circ f\left(i_R,\bar{i}_{\setminus R}\right)$. Let us assume that $h_R$ is not a constant, namely there exist $i^1_R$, $i^2_R$ such that $a^1_R=h_R(i^1_R)\not=h_R(i^2_R)= a^2_R$. We define then the function $g_R:\mathcal{I}_R\rightarrow \mathcal{I}_R$ as
\begin{align}
g_R(a^1_R)&=i^2_R\,, \\
g_R(i_R)&= i^1_R \quad \forall\, i_R \not= a^1_R\,.
\end{align}
It is then easy to see that $h_R\circ g_R$ has no fixed point. Indeed, $h_R\circ g_R (a^1_R)= h_R(i^2_R)=a^2_R\not=a^1_R$, while for any $i_R \not= a^1_R$, $h_R\circ g_R (i_R) = h_R(i^1_R) =a^1_R \not= i_R$. Thus, if $w_R\circ f$ is not a constant over $\mathcal{I}_R$, then $w_R\circ (f\circ g_R)$ has no fixed point and $w_R$ is not a component of a process function. As this must hold for every set of local operations $f$ and every $\bar{i}_{\setminus R}$, we conclude that each component $w_R$ of a process function must be a constant over $\mathcal{O}_R$.
\end{proof}
Lemma~\ref{constant} immediately implies a characterisation of single-region process functions:
\begin{corollary}\label{single}
	Given a function $w:\mathcal{O}\rightarrow\mathcal{I}$, $w\circ f$ admits a fixed point for every $f:\mathcal{I}\rightarrow\mathcal{O}$ ($w$ is a single-region process function) if and only if $w$ is a constant.
\end{corollary}
\begin{proof}
As a consequence of Lemma~\ref{constant}, every single-partite process function must be a constant. On the other hand, for a constant function $w$, $i=w(o)$ is a fixed point of $w\circ f$ for every $f$, thus every constant $w$ is a process function.
\end{proof}

Intuitively, if we fix the operation performed in one of the $N$ regions, we should obtain a process for the remaining $N-1$ regions. This intuition will also play an important role in the proofs below. The first step to formalise such an intuition is the following definition:
\begin{definition}
Consider a function \mbox{$w:\mathcal{O}\rightarrow \mathcal{I}$}, such that, for each region $R=1,\dots,N$, \mbox{$w_R(o)=w_R(o_{\setminus R})$}. For a given local operation $f_R:\mathcal{I}_R\rightarrow \mathcal{O}_R$, we define the \textbf{reduced function} $w^{f_R}:\mathcal{O}_{\setminus R}\rightarrow \mathcal{I}_{\setminus R}$ on the remaining regions by composing $w$ with $f_R$:
\begin{equation}\label{reduced}	w^{f_R}_S\left(o_{\setminus R}\right):= w_S\left(o_{\setminus R},f_R\left(w_R\left(o_{\setminus R}\right)\right)\right)\,, \; S\not=R\,.
\end{equation}
\end{definition}

We will need the following fact:
\begin{lemma}\label{reducedpoint}
If $i\in \mathcal{I}$ is a fixed point of $w\circ f$, then $i_{\setminus R}$ is a fixed point of $w^{f_R}\circ f_{\setminus R}$.
\end{lemma}
\begin{proof}
Since $i$ is a fixed point of $w\circ f$, we have $i_R = w_R(f_{\setminus R}(i_{\setminus R}))$. Then, for $S\not = R$, Eq.~\eqref{reduced} implies
\begin{multline}
	w_S^{f_R}\circ f_{\setminus R}\left(i_{\setminus R} \right)  = w_S\left(f_{\setminus R}(i_{\setminus R}), f_R(i_R)\right)\\
	=w_S\circ f(i)
= i_S.
\end{multline}
\end{proof}
We can now prove two crucial properties.
\begin{lemma} \label{lemma}
	Given a function \mbox{$w:\mathcal{O}\rightarrow \mathcal{I}$}, such that, for each region $R=1,\dots,N$, \mbox{$w_R(o)=w_R(o_{\setminus R})$}, we have
\begin{enumerate}[(i)]
	\item If $w$ is a process function, then $w^{f_R}$ is also a process function for every region $R$ and operation $f_R$.
	\item If there exists a region $R$ such that, for every local operation $f_R$, $w^{f_R}$ is a process function, then $w$ is also a process function.
\end{enumerate}
\end{lemma}
\begin{proof}
Point \emph{(i)} is a direct consequence of Lemma~\ref{reducedpoint}: For every set of operations $f_{\setminus R}$, a fixed point of $w_S^{f_R}\circ f_{\setminus R}$ is given by $i_{\setminus R}$, where $i$ is a fixed point of $f=f_R\cup f_{\setminus R}$. 

To prove \emph{(ii)}, we can set $R=1$ without loss of generality. We then have to prove that, if $w^{f_1}$ is a process function for every $f_1$, it follows that $w$ is also a process function, \ie, we have to find a fixed point of $w\circ f$ for an arbitrary $f$.
As the reduced function $w^{f_1}$ is a process function by assumption, there exists a fixed point $i_{\setminus 1}$ of $w^{f_1}\circ f_{\setminus 1}$. Choosing $i_1=  w_1 \circ f_{\setminus 1}\left(i_{\setminus 1}\right)$ as input state for region $1$, we see that $i:= i_1 \cup i_{\setminus 1}$ is a fixed point of $w\circ f$. Indeed this is true, by definition of $i_1$, for the component $w_1$. For $S>1$, the definition of reduced function, Eq.~\eqref{reduced}, gives
\begin{multline*}
	w_S\circ f \left(i\right) = w_S\left(f_1\left(i_1\right),f_{\setminus 1}\left(i_{\setminus 1}\right)\right)
	\\ = w^{f_1}_S\circ f_{\setminus 1}\left(i_{\setminus 1}\right) = i_{S}
	\,,
\end{multline*}
where in the last equality we used the fact that $i_{\setminus 1}$ is the fixed point of $w^{f_1}_S\circ f_{\setminus 1}$.
\end{proof}

\section{Uniqueness of the fixed point}\label{uniquefp}
Here we prove Theorem~\ref*{fixedthm}, namely the following $N$-dependent proposition.

\emph{$P[N]$: Let $w:\mathcal{O}\rightarrow \mathcal{I}$ be such that, for every collection of functions $f=\{f_1,\dots,f_N\}$, $f_R: \mathcal{I}_R\rightarrow \mathcal{O}_R$, there exists at least one fixed point $i\in \mathcal{I}$, $w\circ f(i)=i$. Then, the fixed point is unique for each $f$.
}

We prove this by induction, namely we first prove $P[1]$ and then the implication $P[N-1]\Rightarrow P[N]$ for $N>1$.

\begin{proof}
$P[1]$ is a simple consequence of Corollary~\ref{single}: $w\circ f$ can have a fixed point for every $f$ only if $w$ is constant, $w(o)=\bar{i}$ for every $o$. Then, $\bar{i}$ is the unique fixed point of $w\circ f$ for every $f$.

For $N>1$, let us assume $P[N]$ is false and let $a=\left\{a_1,\dots,a_N\right\}$, $b=\left\{b_1,\dots,b_N\right\}$ be two distinct fixed points of $w\circ f$, where $w$ is an $N$-partite process function. Without loss of generality, we can assume that they differ in the first component, $a_1\not= b_1$. This means that the reduced function $w^{f_N}\circ f_{\setminus N}$ has two distinct fixed point, $a_{\setminus N}\not=b_{\setminus N}$. But this is in contradiction with $P[N-1]$, because, according to point \textit{(i)} of Lemma~\ref{lemma}, $w^{f_N}$ is an $N-1$-partite process function and thus has a single fixed point for each $f_{\setminus N}$. This concludes the proof. 
\end{proof}

\section{Reversibility}\label{reversibility}
Here we prove that every process function can be extended to an invertible process function (Theorem~\ref*{thm:go}).

\begin{proof}

Given a process function $w:\mathcal{O}\rightarrow \mathcal{I}$ over $N$ local regions, we add two extra regions, $\source$ and $\sink$, a source and a sink, respectively.	We take the output space of the source to be isomorphic with the entire input space of the~$N$ regions, $e\in\mathcal{O}_{\source}\cong \mathcal{I}$, while the input space of the sink is isomorphic to the output space of the regions, $s\in\mathcal{I}_{\sink}\cong \mathcal{O}$. For each $R=1,\dots,N$ and each $e_R\in\mathcal{I}_R$, we introduce a function $T^{e_R}_R:\mathcal{I}_R\rightarrow \mathcal{I}_R$ such that there exists $\tilde{e}_R$ for which $T^{\tilde{e}_R}_R(i_R)=i_R$ and, for each $i_R\in \mathcal{I}_{R}$, $T^{(\cdot)}_R(i_R)$ is invertible. For a state space with a linear structure, we can take $T^{e_R}_R(i_R)=i_R+e_R$ for concreteness (although the proof does not rely on this).\footnote{If $\left|\mathcal{I}_R\right|=c_R<\infty$, we can use $T^{e_R}_R=i_R \oplus e_R$, where $\oplus$ is addition modulo $c_R$.} 
We then extend the process function~$w:\mathcal{O}\rightarrow \mathcal{I}$ to a function~$w'\equiv (w^1,w^2):\mathcal{O}\times\mathcal{O}_{\source}\rightarrow\mathcal{I}\times\mathcal{I}_{\sink}$, defined as
\begin{align}
	 w^1_R&(o,e)= T^{e_R}_R\circ w_R(o) = w_R(o)+e_R  \\
	 w^2_R&(o,e)=o_R.
	\label{extension}
\end{align}

The function $w'$ is invertible, with the inverse given by $i_R=w_{R}(s)-e_R$, $o_R=s_R$. To show that it is a process function, we have to prove that its composition with arbitrary local operations has a fixed point, condition~(\ref*{fixedpoint}) in the main text. Note that this condition is equivalent to the existence of output fixed points: $f'\circ w'(o')=o'$, $o'\equiv (o,e)$. Since local operations for $\source$ are functions $\varnothing\rightarrow \mathcal{O}_{\source}$, where~$\varnothing$ is the empty set, they can be identified with a state $f(\varnothing)\equiv e\in \mathcal{O}_{\source}$, interpreted as `state preparation.' The fixed-point condition for the source components is then trivially satisfied by any $e\in \mathcal{O}_{\source}$. As the sink has no output space, the fixed-point condition reduces to 
\begin{align}
	o_R = f_R\circ T_R^{e_R}\circ w_R(o),
	\label{extension2}
\end{align}
which should be satisfied for every $f\in \mathcal{D}$ and $e\in \mathcal{O}_{\source}$. This is true because $f_R\circ T_R^{e_R}$ is a local operation and, as $w$ is a process function, a fixed point $o\in \mathcal{O}$ exists for every local operation.
\end{proof}

\section{Characterizations}\label{characterizations}
Here we prove the characterisations of process functions for up to three local regions. The single-region characterisation is given by Corollary~\ref{single}: All and only constants are single-region process functions.

\subsection{Two regions}
We relabel input and output of region $R$ as \mbox{$i_R\rightarrow a\in \mathcal{A}$}, \mbox{$o_R\rightarrow x\in \mathcal{X}$}, respectively, and inputs and outputs of region $S$ as $i_S\rightarrow b\in \mathcal{B}$, $o_S\rightarrow y\in \mathcal{Y}$. For a bipartite process function $w=\left\{w_R,w_S\right\}$, the single-party characterisation implies that $w_R(x,y)=w_R(y)$, $w_S(x,y)=w_S(x)$. It is furthermore clear that, if at least one of the two components of a function $w=\left\{w_R,w_S\right\}$ is a constant, then $w$ is a process function. (Given $w_R(y)=a_0$, the fixed point $i\equiv\left\{a,b\right\}$ is given by $a=a_0$, $b=w_S(f_R(a_0))$.) 

It remains to prove that if $w$ is a process function, then at least one of the two components is constant.
\begin{proof}
The consistency condition~(\ref*{fixedpoint}) in the main text says that, for every local operation $f_R$, $f_S$, there exists $a\in \mathcal{A}$, $b\in \mathcal{B}$ such that 
\begin{align}
	w_R\left(f_S(b)\right)=a\,, \qquad w_S\left(f_R(a)\right)=b\,.
\end{align}
By plugging the second equation into the first we obtain
\begin{equation}
	w_R\circ f_S\circ w_S\circ f_R(a)=a\,.
\end{equation}
The single-party characterisation tells us that $w_R\circ f_S\circ w_S$ must be a constant, and this must be true for all $f_S$. This is only possible if one of the two functions, $w_R$, $w_S$, is constant.
\end{proof}

\subsection{Three regions}
We prove here Theorem~\ref*{3charact}, which characterises tripartite process functions as those with one-way conditional signalling.
\begin{proof}
As in the main text, we consider three regions $R$, $S$, $T$, with input states $a\in \mathcal{A}$, $b\in \mathcal{B}$, $c\in \mathcal{C}$ and output states $x\in \mathcal{X}$, $y\in \mathcal{Y}$, $z\in \mathcal{Z}$.
The function $w^z:\mathcal{X}\times \mathcal{Y}\rightarrow \mathcal{A}\times \mathcal{B}$, defined as $w^z(x,y):=\left\{w_{R}(y,z),w_{S}(x,z)\right\}$, is the reduced function obtained from $w$ by fixing the local operation of $T$ to be the constant function with output $z$. If $w$ is a process function, by point \emph{(i)} of Lemma~\ref{lemma}, $w^z$ is a bipartite process function, and similarly for the analogously defined $w^x$ and $w^y$. This proves the `easy' direction of the theorem.

We want to prove the converse: If $w^x$, $w^y$, and $w^z$ are bipartite process functions for arbitrary $x$, $y$, $z$, then $w$ is also a process function. Note that we cannot apply point \emph{(ii)} of Lemma~\ref{lemma} directly, because that requires knowing that the reduced function is a process function for an {\em arbitrary\/} local operation, not just for the constant operation.

By assumption, we know that $w^z$ is a bipartite process function. This means that, for every given $z$, $w^z$ is one-way signalling. In other words, at least one of the two components, $w_R^z\equiv w_{R}(\cdot,z)$ or $w_S^z\equiv w_S(\cdot,z)$, must be a constant. We denote $\mathcal{Z}_R\subset \mathcal{Z}$ the subset of outputs of $T$ for which $w_R^z$ is constant, and $\mathcal{Z}_S\subset \mathcal{Z}$ the subset for which $w_S^z$ is constant. Because at least one of the two components is constant, we have $\mathcal{Z}= \mathcal{Z}_R\cup \mathcal{Z}_S$. (The two subsets can have non-null intersection.) In a similar way, we write $\mathcal{X}= \mathcal{X}_S\cup \mathcal{X}_T$ and $\mathcal{Y}= \mathcal{Y}_R\cup \mathcal{Y}_T$, where $w^x_S$ is constant for $x\in\mathcal{X}_S$, and so on. Thus we have
\begin{align}
	w_R(y_R,z) &= w_R(y,z_R) = a_0 \\ \nonumber
	 \forall&\, y_R \in \mathcal{Y}_R, y \in \mathcal{Y}, z_R \in \mathcal{Z}_R, z \in \mathcal{Z}\,; \\
	w_S(x_S,z) &= w_S(x,z_S) = b_0 \\ \nonumber
	 \forall&\, x_S \in \mathcal{X}_S, x \in \mathcal{X}, z_S \in \mathcal{Z}_S, z \in \mathcal{Z}\,; \\
	w_T(x_T,y) &= w_T(x,y_T) = c_0 \\ \nonumber
	 \forall&\, x_T \in \mathcal{X}_T, y \in \mathcal{Y}, y_T \in \mathcal{Y}_T, x \in \mathcal{X}\,.
\end{align}
Now consider an arbitrary local operation $f_R:\mathcal{A}\rightarrow\mathcal{X}$. We want to show that the reduced function $w^{f_R}$, defined as in Eq.~\eqref{reduced}, is a bipartite process function. To this end, we need to prove:
\begin{enumerate}[(i)]
	\item $w_S^{f_R}(y,z)= w_S^{f_R}(z)\,$,
	\item $w_T^{f_R}(y,z)= w_T^{f_R}(y)\,$,
	\item At least one of the two components is constant.
\end{enumerate}
Let us start with point (i). For $z_S\in \mathcal{Z}_S$, we have $w_S^{f_R}(y,z_S)=b_0$, independently of $y$. Let us then consider $z_R\in \mathcal{Z}_R$. By definition, $w_S^{f_R}(y,z_R)= w_S\left(f_R\circ w_R(y,z_R),z_R\right)$. But $w_R(y,z_R)=a_0$ for $z_R\in \mathcal{Z}_R$, thus $w_S^{f_R}(y,z_R)= w_S\left(f_R(a_0),z_R\right)$, which is again independent of $y$. By a similar argument, $w_T^{f_R}(y,z)$ is independent of $z$.

We are thus left with proving point (iii). We shall prove the equivalent implications
\begin{align*}
w_S^{f_R}\; \textrm{not constant} \Rightarrow w_T^{f_R}\; \textrm{constant}, \\
w_T^{f_R} \textrm{not constant} \Rightarrow w_S^{f_R}\; \textrm{constant}.
\end{align*}
Say that $w_S^{f_R}$ is not constant. Then, there is a $z_R\in \mathcal{Z}_R$ such that $w_S^{f_R}(z_R)\not = b_0$. As we have seen above, for $z_R\in \mathcal{Z}_R$, $w_S^{f_R}(z_R)= w_S\left(f_R(a_0),z_R\right)$, so we need $f_R(a_0)\notin \mathcal{X}_S$ to have $w_S^{f_R}(z_R)\not = b_0$. This means that $f_R(a_0)\in \mathcal{X}_T$. But then, for $y_R\in \mathcal{Y}_R$ we have $w_T^{f_R}(y_R)= w_T\left(f_R(a_0),y_R\right)=c_0$, and also $w_T^{f_R}(y_T)=c_0$ for $y_T\in \mathcal{Y}_T$. This means that $w_T^{f_R}$ is constant. 

To recapitulate, we have proven that, if $w^x$, $w^y$, $w^z$ are bipartite process functions for arbitrary $x\in \mathcal{X}$, $y\in \mathcal{Y}$, $z\in \mathcal{Z}$ (as per hypothesis), then $w^{f_R}$ is a bipartite process function for an arbitrary operation $f_R$. Point \emph{(ii)} of Lemma~\ref{lemma} finally implies that $w$ is a process function, concluding the proof.
\end{proof}

\end{appendix}

\bibliography{refs}

\end{document}